\documentclass[12pt]{article}

\usepackage{amsmath,amssymb,amsfonts,amsthm}
\usepackage{inputenc}
\usepackage{enumerate}
\usepackage{graphicx}
\usepackage{multicol}
\usepackage{multirow}
\newtheorem{thm}{Theorem}
\newtheorem{prop}{Proposition}
\newtheorem{lem}{Lemma}
\newtheorem{cor}{Corollary}

\newtheorem{exam}{Example}

\newtheorem{rem}{Remark}

\def\g{{\mathbf g}}

\def\0{{\mathbf 0}}

\newcommand{\F}{\mathbb{F}}
\newcommand{\Z}{\mathbb{Z}}

\begin{document}

%
%
%

\title{New families of self-dual codes}

\author{
Lin Sok\thanks{ School of Mathematical Sciences, Anhui University, Hefei, Anhui, 230601, { \tt soklin\_heng@yahoo.com}}
}

\date{}
\maketitle
\begin{abstract}
In the recent paper entitled ``Explicit constructions of MDS self-dual codes" accepted in { IEEE Transactions on Information Theory}, doi: 10.1109/TIT.2019.2954877, the author has constructed families of MDS self-dual codes from genus zero algebraic geometry (AG) codes, where the AG codes of length $n$ were defined using two divisors $G$ and $D=P_1+\cdots+P_n.$ In the present correspondence, we explore more families of optimal self-dual codes from AG codes. New families of MDS self-dual codes with odd characteristics and those of almost MDS self-dual codes are constructed explicitly from genus zero and genus one curves, respectively. More families of self-dual codes are constructed from algebraic curves of higher genus.

\end{abstract}
\section{Introduction}
Self-dual codes are one of the most interesting classes of linear codes that find diverse applications in cryptographic protocols (secret sharing schemes) introduced in \cite{Cram, DouMesSol,Massey} and combinatorics \cite{MacSlo}. 
It is well-known that binary self-dual codes are asymptotically good~\cite{MacSloTho}. 

MDS codes form an optimal family of classical codes. They are closely related to combinatorial designs \cite[p. 328]{MacSlo}, and finite geometries \cite[p. 326]{MacSlo}. Due to their largest error correcting capability for given length and dimension, MDS codes are of great interest in both theory and practice. The most well-known family of MDS linear codes is that of Reed-Solomon codes. MDS linear codes exist in a very restrict condition on their lengths as the famous MDS conjecture states: for every linear
$[n, k, n -k + 1]$ MDS code over $\F_q,$ if $1 < k < q,$ then $n \le q + 1,$ except when $q$ is even and $k = 3 $ or $k = q -1,$
in which cases $n \le q + 2.$ The conjecture was proved by Ball \cite{Ball} for $q$ a prime. However, for self-dual case, the conjecture may not be true.

Due to the reasons mentioned above, MDS self-dual codes have been of much interest to many researchers. As we have already known that determining the parameters of a given linear code is a challenging problem in coding theory. However, the parameters of an MDS self-dual code are completely determined by its length. Constructions of MDS self-dual codes are valuable. For classical constructions of MDS self-dual codes, we refer to \cite{BetGeoMas,GeoKou,KimLee,Gue}. Existing families of MDS self-dual codes can be described as follows. 
Grassl {\em et al.} \cite{GraGul} constructed MDS codes of all lengths over ${\mathbb F}_{2^m}$  and of all highest possible length over finite fields of odd characteristics. Jin {\em et al.} \cite{JinXin} proved the existence of MDS self-dual codes over ${\mathbb F}_q$ in odd characteristic for $q\equiv 1 \pmod 4$ and for $q$ a square of a prime for some restricted lengths. Using the same technique developed in \cite{JinXin}, more families of MDS self-dual codes have been constructed in \cite{Yan,FangFu}.  Tong {\em et al.} \cite{TongWang} gave constructions of MDS Euclidean self-dual codes through cyclic duadic codes. The families of known MDS self-dual codes are summarized in Table \ref{table:1}.

\begin{table}
\caption{Existing families of MDS self-dual codes, $\eta$: the quadratic character of $\F_q$}
{\scriptsize
$$
\begin{array}{llc}

q&n&\text{References}\\
\hline\\
q=2^m& n\le q&\\
q=p^m,p\text{ odd prime }& n= q+1&\cite{GraGul}\\
\hline
q=r^2&n\le r&\\
q=r^2,r\equiv 3 \pmod 4&n=2tr,t\le (r-2)/2&\cite{JinXin}\\
\hline
q\equiv 3\pmod 4&n\equiv 3\pmod 4,(n-1)|(q-1)&\\
q\equiv 1 \pmod 4&(n-1)|(q-1)&\cite{TongWang}\\
\hline
q \equiv 1 \pmod4 &n|(q - 1), n < q - 1 &\\
q \text{ odd } &(n - 1)|(q - 1),\eta(1 - n) = 1 &\\
q \text{ odd } &(n - 2)|(q - 1),\eta(2 - n) = 1 &\\
q = r^2, r \text{ odd } &n = tr, t \text{ even }, 1 \le t \le r &\\
q = r^2, r \text{ odd } &n = tr + 1, t\text{ odd }, 1 \le t \le r &\\
q = r^s , r \text{ odd }, s \ge  2 &n = lr, l\text{ even }, 2l|(r - 1) &\cite{Yan}\\
q = r^s , r \text{ odd }, s \ge  2 &n = lr, l\text{ even }, (l - 1)|(r - 1),\eta(1 - l)=1 &\\
q = r^s , r \text{ odd }, s \ge  2 &n = lr + 1, l\text{ odd } , l|(r - 1),\eta(l) = 1 &\\
q = r^s , r \text{ odd }, s \ge  2 &n = lr + 1, l\text{ odd } , (l - 1)|(r - 1),\eta(l - 1) =\eta(-1) = 1 &\\
q = p^m , p\text{ odd prime }  &n= pr + 1, r|m &\\
q = p^m , p\text{ odd prime }  &n= 2p^e, 1 \le e < m,\eta(-1) = 1 &\\
\hline
q=p^m&n|(q-1), (q-1)/n\text{ even }&\\
q=p^m,m\text{ even },r=p^s,s|\frac{m}{2}&n=2tr^\ell,0\le \ell<m/s,1\le t\le (r-1)/2&\\
q=p^m,q\equiv 1\pmod 4& n=2p^\ell,0\le \ell <m&\\
q=p^m,m\text{ even },r=p^s,s|\frac{m}{2}&n=(2t+1)r^\ell+1,0\le \ell <m/s,0 \le t \le (r -1)/2 \text{ or } (\ell,t)=(m/s,0)&\cite{FangFu}\\
q=p^m,q\equiv 1 \pmod 4 &n=p^\ell+1,0\le \ell\le m\\
q=p^m &(n-2)|(q-1),\eta(2-n)=1\\
\hline
q=p^m& n=n_0, p|n_0, (n_0-1)|(q-1)\\
q=p^m,  q\equiv 1\pmod 4&  n=n_0+1,p|n_0, (n_0-1)|(q-1)\\
q=p^m,  q\equiv 1\pmod 4 &    n=p^r+1,1\le r\le m, r|m\\
 q=p^m, q \text{ a square }& n=n_0, (n_0-1)|(q-1) \\
q=p^m, q \text{ a square }&  n=n_0+1, (n_0-1)|(q-1)\\
q=p^m, q \text{ a square }&   n=2n_0,n_0  \text{ odd }, (n_0-1)|(q-1)\\
q=p^m  & n=n_0,n_0|\frac{(q-1)}{2} \\
q=p^m,  q\equiv 1\pmod 4 &   n=n_0+1,  n_0|\frac{(q-1)}{2}\\
q=p^m,q\equiv 1\pmod 4&  n=2p^r,1\le r< m, r|m&\cite{Sok}\\
q=p^m,m=2m_0& n=(t+1)n_0+2, n_0=p^{m_0}-1, n_0\equiv 0\pmod 4,  t  \text{ odd } ,  0< t\le \frac{n_0}{2}+1\\
q=p^m,m=2m_0& n=(t+1)n_0+2, n_0=p^{m_0}-1, n_0\equiv 2\pmod 4,  0< t\le \frac{n_0}{2}\\
q=p^m, q \text{ a square } &n=(t+1)n_0+2, n_0=\frac{q-1}{p^r+1}\text{ even}, 1\le r< m, \frac{n_0(p^r+1)}{2(p^r-1)}  \text{ odd },    t  \text{ odd } ,  1\le t\le p^r\\ 
q=p^m,  q \text{ a square } & n=(t+1)n_0+2,n_0=\frac{q-1}{p^r+1}\text{ even}, 1\le r< m, \frac{n_0(p^r+1)}{2(p^r-1)}  \text{ even },    1\le t\le p^r\\
q=p^m,   q \text{ a square } & n=(t+1)n_0+2, n_0=\frac{q-1}{p^r-1}\text{ even}, 1\le r< m,r|\frac{m}{2},  1\le t\le p^r-2\\

\end{array}
$$
}
\label{table:1}
\end{table}
 
The discovery of algebraic geometry codes in 1981 was due to Goppa \cite{Goppa}, where they were also called geometric Goppa codes. Goppa showed in his paper \cite{Goppa} how to construct linear codes from algebraic curves over a finite field. Despite a strongly theoritical construction, algebraic geometry (AG) codes have asymptotically good parameters, and it was the first time that linear codes improved the so-called Gilbert-Vasharmov bound. Self-dual AG codes were studied by Stichtenoth \cite{Stich88} and  Driencourt et {\em al.} \cite{DriSti}, where they first characterized such codes. However, the construction of MDS self-dual AG codes with odd characteristics or almost MDS self-dual AG codes  was not considered there. 

On the contrary to the MDS case, almost MDS codes exist more frequently, and it is thus worth exploring families of self-dual codes in such a case and those of optimal self-dual codes. In \cite{Stich88}, Stichtenoth gave constructions of self-orthogonal AG codes (and self-dual AG codes for some special cases) but did not consider an embedding the self-orthogonal codes into the self-dual ones.

In this paper, we will discover more families of optimal self-dual codes from algbraic curves over finite fields. Optimal self-dual codes are constructed from rational points on the curves and embedding their orthogonal subcodes. We improve the construction \cite{Sok} and other known constructions over $\F_q$ with $q$ a prime (see Theorem \ref{thm:6}) and also give explict constructions of the cosets of $\F_q$ with desired properties (see Lemma \ref{lem:alpha_ij} and Lemma \ref{lem:alpha_ij_2}). Due to Lemma \ref{lem:embedding2-odd-g} and Lemma \ref{lem:embedding2-even-g}, new classes of self-dual codes with prescribed minimum distance are constructed. Additionally, we construct MDS self-dual codes with new parameters 
$[24,12,13]_{37}$, $ [32,16,17]_{41}$, $[26,13,14]_{61}$, $[42,21,22]_{61}$, $[50,25,26]_{73}$, $[24,12,13]_{81},$ almost MDS self-dual codes with new parameters $[16,8,8]_9,[16,8,8]_{16}$, $[18,9,9]_{16}$, $[20,10,10]_{16}$, $[22,11,11]_{16}$, $[24,12,12]_{16}$ and optimal self-dual codes with new parameters $[28,14,12]_9,[26,13,12]_{16},[28,14,13]_{16},[30,15,14]_{16},[32,16,15]_{16}.$

The paper is organized as follows: Section \ref{section:pre} gives preliminaries and background on algebraic geometry codes. Section \ref{section:sd-AGcode} provides explicit constructions of self-dual codes from various algebraic curves. New families of self-dual codes are presented as well some numerical examples are also given. We end up with concluding remark in Section \ref{section:conclusion}.

\section{Preliminaries}\label{section:pre}

Let $\F_q$ be the finite field with $q$ elements. A 
linear code of length $n$ and dimension $k$ over ${{\mathbb F}_q},$ denoted as $q$-ary $[n,k]$ code, is a $k$-dimensional subspace
of  $ {\mathbb F}_q^n$. The
(Hamming) weight wt$({\bf{x}})$ of a vector ${\bf{x}}=(x_1, \dots,
x_n)$ is the number of nonzero coordinates in it. The {\em minimum
distance ({\rm{or}} minimum weight) $d(C)$} of $C$ is
$d(C):=\min\{{\mbox{wt}}({\bf{x}})|{\bf{x}} \in C, {\bf{x}} \ne
{\bf{0}} \}$. The parameters of an $[n,k]$ code with minimum distance $d$ are written $[n,k,d]$. 
If $C$ is an $[n,k,d]$ code, then from the Singleton bound, its minimum distance is bounded above by
$$d\le n-k+1.$$
A code meeting the above bound is called {\em Maximum Distance Separable} ({MDS}). A code is called {\em almost} MDS if its minimum distance is one unit less than the MDS case. A code is called {\it optimal} if it has the highest
possible minimum distance for its length and dimension.
The {\em Euclidean inner product} of ${\bf{x}}=(x_1,
\dots, x_n)$ and ${\bf{y}}=(y_1, \dots, y_n)$ in ${\mathbb F}_q^n$ is
${\bf{x}}\cdot{\bf{y}}=\sum_{i=1}^n x_i y_i$. The {\em dual} of $C$,
denoted by $C^{\perp}$, is the set of vectors orthogonal to every
codeword of $C$ under the Euclidean inner product. A linear code
$C$ is called 
{\em self-orthogonal} if $C\subset C^{\perp}$ and 
{\em self-dual} if $C=C^{\perp}$. It is well-known that a self-dual code can only exist for even lengths.

We refer to Stichtenoth \cite{Stich} for undefined terms related to algebraic function fields. 

Let ${\cal X}$ be a smooth projective curve of genus $g$ over $\F_q.$
The field of rational functions of ${\cal X}$ is denoted by $\F_q({\cal X}).$ Function fields of
algebraic curves over a finite field can be characterized as
finite separable extensions of $\F_q(x)$. We identify points on the curve ${\cal X}$ with places of the
function field $\F_q({\cal X}).$ A point on ${\cal X}$ is called rational if all of
its coordinates belong to $\F_q.$ Rational points can be identified
with places of degree one. We denote the set of $\F_q$-rational
points of ${\cal X}$ by  ${\cal X}(\F_q)$.

A divisor $G$ on the curve ${\cal X}$ is a formal sum $\sum\limits_{P\in {\cal X}}n_PP$ with only finitely many nonzeros $n_P\in \Z$.
The support of $G$ is defined as $supp(G):=\{P|n_P\not=0\}$. The degree of $G$ is defined by $\deg(G):=\sum\limits_{P\in {\cal X}}n_P\deg(P)$. 
For two divisors $G=\sum\limits_{P\in {\cal X}}n_PP$ and $H=\sum\limits_{P\in {\cal X}}m_PP$,  we say that $G\ge H$ if $n_P\ge m_P$ for all places $P\in {\cal X}$.

It is well-known that a nonzero polynomial $f(x)\in \F_q(x)$ can be factorized into irreducible factors as $f(x)=\alpha \prod\limits_{i=1}^s p_i(x)^{e_i},$ with $\alpha\in \F_q^*.$ Moreover, any irreducible polynomial $p_i(x)$ corresponds to a place, say $P_i$. We define the valuation of $f$ at $P_i$ as $v_{P_i}(f):=t$ if $p_i(x)^t|f(x)$ but $p_i(x)^{(t+1)}\not|f(x).$ 

For a nonzero rational function $f$ on the curve $\cal X$, we define the ``principal" divisor of $f$ as
$$(f):=\sum\limits_{P\in {\cal X}}v_P(f)P.$$ 

If $Z(f)$ and  $N(f)$ denotes the set of zeros and poles of $f$ respectively, we define the zero divisor and pole divisor of $f$, respectively by

$$
\begin{array}{c}
(f)_0:=\sum\limits_{P\in Z(f)}v_{P}(f)P,\\
(f)_\infty:=\sum\limits_{P\in N(f)}-v_{P}(f)P.\\
\end{array}
$$
Then $(f)=(f)_0-(f)_\infty $, and it is well-known that the principal divisor $f$ has degree $0.$

We say that two divisors $G$ and $H$ on the curve $\cal X$ are equivalent if $G=H+(z)$ for some rational function $z\in \F_q({\cal X}).$

For a divisor $G$ on the curve $\cal X$, we define
$${\cal L}(G):=\{f\in \F_q({\cal X})\backslash \{0\}|(f)+G\ge 0\}\cup \{0\},$$ 
and 
$${\Omega}(G):=\{\omega\in \Omega\backslash \{0\}|(\omega)-G\ge 0\}\cup \{0\},$$
where $\Omega:=\{fdx|f\in \F_q({\cal X})\}$, the set of differential forms on $\cal X$. It is well-known that, for a differential form $\omega$ on $\cal X$, there exists a unique a rational function $f$ on $\cal X$ such that $$\omega=fdt,$$
where $t$ is a local uniformizing parameters. In this case, we define the divisor associated to $\omega$ by $$(\omega)=\sum\limits_{P\in {\cal X}}v_P(\omega)P, $$
where $v_P(\omega):=v_P(f).$

Through out the paper, we let $D=P_1+\cdots+P_n$, called the rational divisor, where $(P_i)_{1\le i \le n}$ are places of degree one, and $G$ a divisor with $supp(D)\cap supp(G)=\emptyset$. Define the algebraic geometry code by
$$
C_{\cal L}(D,G):=\{(f(P_1),\hdots,f(P_n))|f\in {\cal L}(G)\},
$$
and the differential algebraic geometry code as
$$
C_{\Omega}(D,G):=\{(\text{Res}_{P_1}(\omega),\hdots,\text{Res}_{P_n}(\omega))|\omega\in {\Omega}(G-D)\},
$$
where $\text{Res}_{P}(\omega)$ denotes the residue of $\omega$ at point $P.$


%

The parameters of an algebraic geometry code $C_{\cal L}(D,G)$ is given as follows.
\begin{thm}\cite[Corollary 2.2.3]{Stich}\label{thm:distance} Assume that $2g -2 < deg(G) < n.$ Then the code $C_{\cal L}(D,G)$  has parameters $[n,k,d]$ satisfying
\begin{equation}
k=\deg (G)-g+1\text{ and } d\ge n-\deg (G).
\label{eq:distance}
\end{equation}

\end{thm}

The dual of the algebraic geometry code $C_{\cal L}(D,G)$ can be described as follows.

\begin{lem}\cite[Theorem 2.2.8]{Stich}\label{lem:dual1} With above notation, the two codes $C_{\cal L}(D,G)$ and $C_{\Omega}(D,G)$ are dual to each other.
\end{lem}

Moreover, the differential code $C_{\Omega}(D,G)$ is determined as follows.
\begin{lem}\cite[Proposition 2.2.10] {Stich}\label{lem:dual2} With the above notation, assume that there exists a differential form $\omega $ satisfying
\begin{enumerate}
\item  $v_{P_i}(\omega)=-1,1\le i \le n$ and 
\item $\text{Res}_{P_i}(\omega)=\text{Res}_{P_j}(\omega)$ for $1\le i \le n.$ 
\end{enumerate}
Then $C_{\Omega}(D,G)=a\cdot C_{\cal L}(D,D-G+(\omega))$ for some $a\in ({\F^*_q}){^n}.$ 
\end{lem}


\section{Self-dual algebraic geometry codes}\label{section:sd-AGcode}
In this section, we will construct self-dual codes from algebraic geometry codes. Self-dual codes can be constructed directly from Lemma \ref{lem:char1} or from their self-orthogonal subcodes by extending the basis of the existing codes.

\begin{lem}\cite[Corollary 3.4]{Stich88}\label{lem:char1} With the above notation, assume that there exists a differential form $\omega $ satisfying
\begin{enumerate}
\item  $v_{P_i}(\omega)=-1, 1\le i \le n $ and 
\item $\text{Res}_{P_i}(\omega)=\text{Res}_{P_j}(\omega)=a_i^2$, $1\le i \le n,$  for some $a_i\in {\F^*_q}.$
\end{enumerate} 
Then the following statements hold.
\begin{enumerate}
\item  If $2G\le D+(\omega),$ then there exists a divisor $G'$ such that $C_{\cal L}(D,G)\sim C_{\cal L}(D,G')$, and  $C_{\cal L}(D,G')$ is self-orthogonal.
\item  If $2G= D+(\omega),$ then there exists a divisor $G'$ such that $C_{\cal L}(D,G)\sim C_{\cal L}(D,G')$, and  $C_{\cal L}(D,G')$ is self-dual.
\end{enumerate}
\end{lem}
The existence of self-dual algebraic geometry codes can be given as follows.

\begin{prop} \cite[Corollary 3.1.49, p.292]{TV} With the above notation, assume that $N=|{\cal X}(\F_q)|>2g.$ Then there exists a self-dual code with parameters $[n,\frac{n}{2},\ge \frac{n}{2}-g+1]$ over $\F_q$ for some $n$ even such that $n\ge N-2g-1.$
\end{prop}

The following lemma will be applied many times for constructing  a $q$-ary self-dual code of length $n$ (if it exists for such a length).
\begin{lem}\label{lem:embedding2-odd} Let $n$ be an odd positive integer and $C$ a $q$-ary self-orthogonal code with parameters $[n,\frac{n-1}{2}]$. Then there exists a self-orthogonal code $C_0$ with parameters $[n+1,\frac{n-1}{2}]$ and a self-dual code $C_0'$ with parameters $[n+1,\frac{n+1}{2}]$ such that $C_0\subset C_0'\subset C_0^\perp.$
\end{lem}
\begin{proof} Let $G$ be the generator matrix of $C$ and $C_0$ be a self-orthogonal code obtained from $C$ by lengthening one zero coordinate. Clearly, the code $C_0$ has parameters $[n+1,\frac{n-1}{2}]$, and $C_0^\perp$ has parameters $[n+1,\frac{n+1}{2}+1]$. Denote $G_0$ the generator matrix of $C_0$, that is, 
$$G_0=
\left(
\begin{array}{cc}
&0\\
G&\vdots\\
&0\\
\end{array}
\right).
$$
Let ${\bf x}$ be a nonzero element in the quotient space $C_0^\perp\slash C_0$ such that ${\bf x}\cdot {\bf x}=0.$ Then the code $C_0'$ with its following generator matrix $G_0'$ is self-dual with parameters $[n+1, \frac{n+1}{2}]:$

$$G'_0=
\left(
\begin{array}{ccc}
&&0\\
G&&\vdots\\
&&0\\
\hline
&{\bf x}&\\
\end{array}
\right).
$$

Moreover, we have the following inclusion
$$C_0\subset C'_0\subset {C_0}^\perp.$$
\end{proof}
Similarly, we have the following embedding.
\begin{lem}\label{lem:embedding2-even} Let $n$ be an even positive integer and $C$ a $q$-ary self-orthogonal code with parameters $[n,\frac{n}{2}-1]$. Then there exists  a self-dual code $C'$ (if it exists for such a length) with parameters $[n,\frac{n}{2}]$ such that $C\subset C'\subset C^\perp.$
\end{lem}

\begin{lem}\label{lem:embedding2-odd-g} Let $\cal X$ be a smooth projective curve having genus $g.$ Let $n$ be an odd positive integer and $D=P_1+\cdots+P_n$ be a divisor on $\cal X$. Assume that there exists a differential form $\omega $ satisfying

\begin{enumerate}
\item $v_{P_i} (\omega) = -1,$ for $i = 1,\hdots, n$ and
\item$ \text{Res}_{P_i} (\omega)  = \text{Res}_{P_j} (\omega) )=a_i^2$ with $a_i\in \F_q^*$ for $1 \le i, j \le n.$
\end{enumerate}
If $G=\frac{(2g-3+n)}{2}P_\infty$ with $\text{supp}(G)\cap \text{supp}(D)=\emptyset ,$ then there exists a self-orthogonal code $C_{\cal L}(D,G)$ with parameters $[n,\frac{n-1}{2},\frac{n+3}{2}-g]$.
Moreover, the code $C_{\cal L}(D,G)$ can be embedded into a self-dual $[n+1,\frac{n+1}{2},\ge \frac{n+1}{2}-g]$ code $C'$(if a self-dual code exists for such a length $n$).

\end{lem}
\begin{proof}Choose $U$ as a subset of $\F_q$ with its size $|U|=n$ so that $\omega=\frac{dx}{h},$ where $h(x)=\prod\limits_{\alpha \in U}(x-\alpha),$ satisfying the above two conditions. Then the divisor $(\omega)=(2g-2+n)P_\infty-D$, and thus $2G\le (\omega)+D$. From Lemma \ref{lem:char1} and Theorem \ref{thm:distance}, there exists a self-orthogonal code $C_{\cal L}(D,G)$ with parameters $[n,\frac{n-1}{2},\ge \frac{n+3}{2}-g]$. The second assertion follows from Lemma \ref{lem:embedding2-odd}. 
First note that
\begin{equation}
\begin{array}{ll}
C_{\cal L}(D,G)^\perp&=a\cdot C_{\cal L}(D,D-G+(\omega))\text{ (from Lemma \ref{lem:dual2})}\\
&=a\cdot C_{\cal L}\left(D,\frac{(2g-1+n)}{2}P_\infty\right)\\
\end{array}
\label{eq:dual}
\end{equation}
We now calculate the lower bound on the minimum distance of the dual code.
$$
\begin{array}{ll}
d(C_{\cal L}(D,G)^\perp)&\ge n-\frac{(2g-1+n)}{2}~ (\text{due to  }(\ref{eq:dual})\text{ and Theorem}~\ref{thm:distance})\\
&=\frac{n+1}{2}-g.
\end{array}
$$
The minumum distance of $C'$ follows from the fact that $C'\subset C^\perp_{\cal L}(D,G)$, and this completes the proof.
\end{proof}

\begin{lem}\label{lem:embedding2-even-g} Let $\cal X$ be a smooth projective curve having genus $g.$ Let $n$ be an even positive integer and $D=P_1+\cdots+P_n$ be a divisor on $\cal X$. Assume that there exists a differential form $\omega $ satisfying

\begin{enumerate}
\item $v_{P_i} (\omega) = -1,$ for $i = 1,\hdots, n$ and
\item$ \text{Res}_{P_i} (\omega)  = \text{Res}_{P_j} (\omega) )=a_i^2$ with $a_i\in \F_q^*$ for $1 \le i, j \le n.$
\end{enumerate}
If $G=\frac{(2g-2+n)}{2}P_\infty$ with $\text{supp}(G)\cap \text{supp}(D)=\emptyset ,$ then there exists a self-orthogonal code $C_{\cal L}(D,G)$ with parameters $[n,\frac{n}{2},\frac{n}{2}+1-g]$.

\end{lem}
\begin{proof}The result follows from the same reasoning as that in Lemma \ref{lem:embedding2-odd-g}.
\end{proof}

\subsection{Self-dual codes from projective lines}
In this subsection, we will discover new families of MDS self-dual codes based on the work from \cite{Sok}. In what follows, we let for $a\in \F_q$, $\eta(a):=1$ if $a$ is a square in $\F_q$, and $\eta(a):=-1$ if $a$ is not a square in $\F_q.$

The following two lemmas \cite{Sok} will be used to construct self-dual codes of genus zero.
\begin{lem}\cite[Lemma 6]{Sok} \label{lem:new-orthogonal} For $G=sP_\infty$ with $s\le \lfloor \frac{n-2}{2}\rfloor $, if $(h'(P_i))_{1\le i \le n}$ are squares in $\F_q^*$, then $C_{\cal L}(D,G-(1/\sqrt{h'}))$ is an MDS self-orthogonal code.
\end{lem}

The following lemma is useful for constructing a self-dual code from its self-orthogonal subcode.
\begin{lem}\cite[Lemma 7]{Sok} \label{lem:embedding1} Let $q\equiv 1\pmod 4$. Assume that $G=(k-1)P_\infty,n=2k+1$, and $(h'(P_i))_{1\le i \le n}$ are squares in $\F_q^*.$  Then the $q$-ary self-orthogonal code $C_{\cal L}(D,G-({1}/\sqrt{h'}))$ with parameters $[n,k]$ can be embedded into a $q$-ary MDS self-dual $[n+1,k+1]$ code.
\end{lem}

Now, we construct MDS self-dual codes from Lemma \ref{lem:new-orthogonal} and Lemma \ref{lem:embedding1}.

\begin{thm}$\label{thm:6}\text{ Let }q=p^m$ be an odd prime power.
If  $\eta(-1)=\eta(n)=1$, $n|(q-1)$ and $n$ even, then there exists an $[2n+2,n+1,n+2]\text{ self-dual code over $\F_{q}$}.$
\end{thm}
\begin{proof} Let $U_{n}=\{\alpha\in \F_q^*| \alpha^{n}=1\}$. Let $\beta_1\in \F_q^*$ such that $\beta_1^{n}-1$ is a nonzero square in $\F_q$. Put $U=U_{n}\cup \beta_1U_{n}\cup\{0\}$, and write
$$h(x)=\prod\limits_{\beta\in U}(x-\beta).$$
Then we have that $$h'(x)=((n+1)x^{n}-1)(x^{n}-\beta_1^{n})+nx^{n}(x^{n}-1).$$

Consider the following quadratic equation 
\begin{equation}\label{eq:quadratic}
a^2+b^2=1.
\end{equation}
For any $q$,  (\ref{eq:quadratic}) has $T=(q-1)-4$ solutions, say $(a_1,\pm b_1),\hdots,(a_{\frac{T}{2}},\pm b_{\frac{T}{2}})$, with $(a_i,b_i)\not=(0,\pm 1),(\pm 1,0)$. Take $\beta_1=\sqrt[n]{a_i^2}$ for some $1\le i\le t,(t< {\frac{T}{2}}).$ Then we get $1-\beta_1^n=1-a_i^2=b_i^2$  which are squares in $\F_q^*.$

We have that $-1, n$ are squares in $\F_q$. Moreover, since $\beta_1^{n}$ and $(\beta_1^{n}-1)$ are squares in $\F_q^*$, it implies that $h'(\beta)$ is a square in $\F_q^*$ for any $\beta \in U.$ Now, the fact that all the roots of $h(x)$ are simple gives rise to a self-orthogonal code with parameters $[2n+1,n,n+2]$. 
 Thus, by Lemma \ref{lem:embedding1}, it can be embedded into a $q$-ary self-dual code with parameters $[2n+2,n+1,n+2].$
\end{proof}
\begin{exam} We construct MDS self-dual codes with new parameters as follows.
\begin{enumerate} 
\item Taking $q=37,n=12,$ we obtain a self-dual code  over $\F_{37}$ with parameters $[26,13,14].$
\item Taking $q=61,n=12,20$ we obtain self-dual code over $\F_{61}$ with parameters $[26,13,14],[42,21,22]$, respectively.
\item Taking $q=73,n=24,$ we obtain a self-dual code over $\F_{73}$ with parameters $[50,25,26].$
\end{enumerate}
\end{exam}

\begin{rem} In the proof of Theorem \ref{thm:6}, we have found many values of $\beta_i$ such that $1-\beta_i^n$ is a square. Furthermore, if there exist $\beta_1$ and $\beta_2$ such that $\beta_1^n-\beta_2^n$ is again a square, then we can construct an MDS self-dual code of length $3n+2$ over $\F_q.$ For example, taking $q=41,n=10$ and considering two non-zero multiplicative cosets of $U_n$ yields a self-dual code over $\F_{41}$ with parameters $[32,16,17].$ The generator matrix of the self-dual code over $\F_{41}$ is given as follows.

{\tiny
$$
\left(
\begin{array}{ccccccccccccccccc}
&20&15&20&11&6&11&22&10&39&15&5&9&20&31&33&40\\
&37&5&11&12&3&15&4&10&37&18&35&8&3&26&32&8\\
&12&29&9&14&13&11&8&30&16&20&17&22&3&9&13&10\\
&2&17&11&25&14&3&1&27&38&5&1&15&36&2&1&21\\
&27&21&21&13&20&7&36&15&29&30&25&20&1&11&2&32\\
&39&7&2&1&26&25&5&38&38&13&33&20&17&15&7&36\\
&40&39&34&15&18&12&6&28&25&10&21&23&8&35&26&26\\
{I_{16}}&30&36&28&2&1&11&12&28&2&27&34&35&4&4&20&2\\
&22&18&5&24&5&40&23&9&34&40&12&34&9&34&33&31\\
&20&13&9&12&31&35&37&33&26&37&23&39&29&18&25&19\\
&11&19&18&16&38&40&2&29&8&30&30&10&12&2&20&30\\
&34&13&10&13&18&28&19&14&28&31&4&34&24&9&31&35\\
&24&31&21&40&12&23&25&4&17&27&13&4&31&40&23&30\\
&40&31&36&35&28&38&21&31&14&20&16&36&20&37&34&21\\
&9&10&23&11&36&23&30&9&16&22&27&32&37&26&39&26\\
&36&4&32&32&4&4&10&14&12&14&20&30&29&34&8&21\\
\end{array}
\right).
$$
}
\end{rem}
The two following lemmas play the key role in determining whether the difference of two special elements in $\F_q$ is a square or not and also in determining the number of cosets of a multiplicate subgroup of $\F^*_q.$
\begin{lem} \label{lem:alpha_ij}Let $q=p^m$ with $p$ an odd prime, $n=\frac{q-1}{p^r+1}$ and for $\alpha_i,\alpha_j\in \F_q$ with $\alpha_i\not =\alpha_j$, denote $\alpha_{ij}=\alpha_i^n-\alpha_j^n.$ Then for $\omega$ a primitive element of $\F_q$, we have the following equality:
\begin{equation}\label{eq:alpha_ij}
\alpha_{ij}=\frac{\omega^{\frac{n(p^r+1)}{2(p^r-1)}}}{\alpha_i\alpha_j}.
\end{equation}
\end{lem}
\begin{proof}Raising $\alpha_{ij}$ to the power $p^r-1,$ we get
$$
\begin{array}{ll}
\alpha_{ij}^{p^r-1}
&=\frac{\left(\alpha_i^n-\alpha_j^n\right)^{p^r}}{\alpha_i^n-\alpha_j^n}
=\frac{(\alpha_i^{n{p^r}}-\alpha_j^{np^r})}{\alpha_i^n-\alpha_j^n}\\
&=\frac{(\alpha_i^{q-1-n}-\alpha_j^{q-1-n})}{\alpha_i^n-\alpha_j^n}
=\frac{\frac{1}{\alpha_i^n}-\frac{1}{\alpha_j^n}}{\alpha_i^n-\alpha_j^n}\\
&=\frac{\omega^{\frac{n(p^r+1)}{2}}}{\alpha_i^n\alpha_j^n},
\end{array}
$$
where the last equality come from the fact that $\omega^{\frac{q-1}{2}}=-1.$

By taking the $(p^r-1)$-th root, the result follows.
\end{proof}

\begin{lem}\label{lem:alpha_ij_2} Let $q=p^m$ with $p$ an odd prime, $r|m$, $n=\frac{q-1}{p^r-1}$ and for $\alpha_i,\alpha_j\in \F_q$ with $\alpha_i\not =\alpha_j$, denote $\alpha_{ij}=\alpha_i^n-\alpha_j^n.$ Then for $\omega$ a primitive element of $\F_q$, we have the following:
\begin{equation}\label{eq:alpha_ij_2}
\alpha_{ij} \in \F_{p^r}.
\end{equation}
\end{lem}
\begin{proof}Raising $\alpha_{ij}$ to the power $p^r,$ we get
$$
\alpha_{ij}^{p^r}={(\alpha_i^{n{p^r}}-\alpha_j^{np^r})}
={(\alpha_i^{q-1+n}-\alpha_j^{q-1+n})}=\alpha_i^n-\alpha_j^n=\alpha_{ij}.
$$
Thus the result follows.
\end{proof}
\begin{thm}\label{thm:new-multicoset} Let $q=p^m$ be an odd square and $n$ even. Put $s=(t+1)n.$

\begin{enumerate}
\item If $\frac{n(p^r+1)}{2(p^r-1)}$ is even, then there exists a self-dual code over $\F_q$ with parameters $[s,\frac{s}{2},\frac{s}{2}+1]$, with $n=\frac{q-1}{p^r+1}$, for $1\le t\le p^r.$
\item If $\frac{n(p^r+1)}{2(p^r-1)}$ is odd, then there exists a self-dual code over $\F_q$ with parameters $[s,\frac{s}{2},\frac{s}{2}+1]$, with $n=\frac{q-1}{p^r+1}$, for $t$ odd and $1\le t\le p^r.$

\item There exists a self-dual code over $\F_q$ with parameters $[s,\frac{s}{2},\frac{s}{2}+1]$, with $n=\frac{q-1}{p^r-1},r|\frac{m}{2}$, for $1\le t\le p^r-2.$
\end{enumerate}

\end{thm}

\begin{proof} Let $U_n$ be a multiplicative subgroup of $\F_q^*$ of order $n$, say $U_n=\{u_1,\hdots,u_n\}.$ Let $\alpha_1U_n,\hdots,\alpha_tU_n$ be $t$ nonzero cosets of $U_n$, where $(\alpha_i)_{1\le i \le t}$ will be determined later. Put $U=U_n\cup \left(\bigcup\limits_{i=1}^t\alpha_iU_n\right),$ and write
$$h(x)=\prod\limits_{\alpha\in U}(x-\alpha).$$ 
Clearly,  all the roots of $h(x)$ are simple. The derivative of $h(x)$ is given by $$h'(x)=nx^{n-1}\prod\limits_{i=1}^t(x^n-\alpha_i^n)+nx^{n-1}(x^n-1)\left(\sum\limits_{i=1}^t\prod\limits_{j=1,j\not=i}^t(x^n-\alpha_j^n)\right).$$

For $1\le j\le t,1\le s\le n$, we have 
$$
\begin{array}{ll}
h'(u_s)&=nu_s^{n-1}(\alpha_j^n-1)\prod\limits_{i=1,i\not=j}^t(1-\alpha_i^n),\\
h'(\alpha_ju_s)&=n(\alpha_ju_s)^{n-1}(\alpha_j^n-1)\prod\limits_{i=1,i\not=j}^t(\alpha_j^n-\alpha_i^n).\\
\end{array}
$$
For $1\le i,j\le t$ and $n=\frac{q-1}{p^r+1},$ we have from  (\ref{eq:alpha_ij})
$$\alpha_{ij}=\alpha_i^n-\alpha_j^n=\frac{\omega^{\frac{n(p^r+1)}{2(p^r-1)}}}{\alpha_i\alpha_j}, $$
where $\omega$ is a primitive element of $\F_q.$

Fixing $j$ and taking all the product of $\alpha_{ij}$ for $1\le i \le t, i\not= j$, we get that

$$\prod\limits_{i=1,i\not=j}^t(\alpha_j^n-\alpha_i^n)=\prod\limits_{i=1,i\not= j}^n\frac{\omega^{\frac{n(p^r+1)}{2(p^r-1)}}}{\alpha_i\alpha_j}.$$

 Obviously, $n$ and $(u_s)_{1\le s\le n}$ are squares in $\F_q$ for $q$ a square. Now, the squareness of $h'(u_s)$ and $h'(\alpha_ju_s)$ depend on the parity of $T={\frac{n(p^r+1)}{2(p^r-1)}}.$  

If $T$ is even, then $\alpha_i$ is chosen to be a square element in $\F_q$, and thus $(1-\alpha_i^n)$ and $(\alpha_j^n-\alpha_i^n)$ are square elements in $\F_q$ due to  (\ref{eq:alpha_ij}) of Lemma \ref{lem:alpha_ij}.

 If $T$ is odd, then $\alpha_i$ is chosen to be a non-square element in $\F_q$, and thus $(1-\alpha_i^n)$ and $(\alpha_j^n-\alpha_i^n)$ are again square elements in $\F_q$ due to  (\ref{eq:alpha_ij}).

In conclusion, we have

\begin{enumerate}
\item If $\frac{n(p^r+1)}{2(p^r-1)}$ is even, then $h'(u_s)$ and $h'(\alpha_ju_s)$ are squares in $\F_q^*$ for  $1\le s\le n$ and $1\le j\le t$ with $t\in \{1,\hdots,p^r\}$.
\item If $\frac{n(p^r+1)}{2(p^r-1)}$ is odd, then $h'(u_s)$ and $h'(\alpha_ju_s)$ are squares in $\F_q^*$ for  $1\le s\le n$ and $1\le j\le t$
with $t$ odd and $t\in \{1,\hdots,p^r\}$.
\end{enumerate}
For $1\le i,j\le t$ and $n=\frac{q-1}{p^r-1},$ from  (\ref{eq:alpha_ij_2}) in Lemma \ref{lem:alpha_ij_2}, we get that $\alpha_{ij}=\alpha_i^n-\alpha_j^n\in\F_{p^r}$, and hence it is a square if $r|\frac{m}{2}.$
We have shown that $h'(\alpha)$ is a nonzero square in $\F_q$ for any $\alpha\in U,$ and thus the constructed code is  self-dual by Lemma \ref{lem:new-orthogonal}.
\end{proof}
\begin{exam}Taking $q=9^2,n=\frac{9^2-1}{9+1}=8,t=2,$ we get an MDS self-dual with parameters $[24,12,13].$ These parameters are new. The generator matrix of the code is given as follows.

{\tiny
$$
\left(
\begin{array}{lllllllllllll} 
&w^{48 }&w^{20 }&w^{72 }&w^{77 }&w^{19 }&w^{16 }&w^{44 }&w^{10 }&w^{37 }&w^{76 }&w^{47 }&w^{56}\\
&w^{20 }&w^{26 }&w^{73 }&w^{56}&{ 2 }&w^{26 }&w^{60 }&w^{55 }&w^{65 }&w^{47 }&w^{41}&w^{47}\\
&w^{32 }&w^{33 }&w^{34 }&w^{14 }&w^{54 }&w^{2 }&w^{7 }&w^{26 }&w^{65 }&w^{75 }&w^{47 }&w^{76}\\
&w^{37 }&w^{16 }&w^{14 }&w^{16 }&w^{53 }&w^{59 }&w^{53 }&w^{43 }&w^{8 }&w^{25 }&w^{25}&w^{77}\\
&w^{19}&{ 2 }&w^{14 }&w^{13 }&w^{4 }&w^{5 }&w^{57 }&w^{9 }&w^{43 }&w^{66 }&w^{15 }&w^{50}\\
I_{12}&w^{16 }&w^{26 }&w^{42 }&w^{19 }&w^{5}&{ 2 }&w^{9 }&w^{17 }&w^{13 }&w^{7 }&w^{60 }&w^{44}\\
&w^{4 }&w^{20 }&w^{7 }&w^{53 }&w^{17 }&w^{49}&{ 2 }&w^{45 }&w^{19 }&w^{2 }&w^{26 }&w^{16}\\
&w^{10 }&w^{55 }&w^{66 }&w^{3 }&w^{9 }&w^{17 }&w^{5 }&w^{44 }&w^{13 }&w^{54}&{ 2 }&w^{19}\\
&w^{37 }&w^{65 }&w^{25 }&w^{48 }&w^{43 }&w^{13 }&w^{59 }&w^{13 }&w^{16 }&w^{54 }&w^{16 }&w^{37}\\
 &w^{36 }&w^{7 }&w^{75 }&w^{25 }&w^{26 }&w^{47 }&w^{2 }&w^{14 }&w^{14 }&w^{74 }&w^{33 }&w^{32}\\
&w^{47 }&w^{41 }&w^{7 }&w^{65 }&w^{15 }&w^{60 }&w^{66}&{ 2 }&w^{16 }&w^{73 }&w^{66}&w^{60}\\
&w^{56 }&w^{47 }&w^{36 }&w^{37 }&w^{50 }&w^{44 }&w^{56 }&w^{19 }&w^{37 }&w^{72 }&w^{60 }&w^{8}\\
\end{array}
\right).
$$
}
\end{exam}

\begin{thm} Let $q$ be an odd prime power with $q\equiv 1\pmod 4, 1\le r<m,r|\frac{m}{2}$ and $n$ even with $\eta(n)=1.$ Put $s=(t+1)n.$

\begin{enumerate}
\item If $\frac{n(p^r+1)}{2(p^r-1)}$ is odd, then there exists a self-dual code over $\F_q$ with parameters $[s,\frac{s}{2},\frac{s}{2}+1]$, with $n=\frac{q-1}{p^r+1}$, for $t$ odd and $1\le t\le p^r.$
\item If $\frac{n(p^r+1)}{2(p^r-1)}$ is even, then there exists a self-dual code over $\F_q$ with parameters $[s,\frac{s}{2},\frac{s}{2}+1]$, with $n=\frac{q-1}{p^r+1}$, for $1\le t\le p^r.$
\item There exists a self-dual code over $\F_q$ with parameters $[s,\frac{s}{2},\frac{s}{2}+1]$, with $n=\frac{q-1}{p^r-1},r|\frac{m}{2}$, for $1\le t\le p^r-2.$
\end{enumerate}

\end{thm}
\begin{proof} The proof follows from that of Theorem \ref{thm:new-multicoset}.
\end{proof}

\subsection{Self-dual codes from elliptic curves and hyper-elliptic curves}

In this subsection, we will consider elliptic curves and hyper-elliptic curves over $\F_q, q$ even.

First, we will consider elliptic curves in Weierstrass form to construct self-dual codes. Let $q = p^m$ and an elliptic curve defined by the equation

\begin{equation}
{\cal E}_{a,b,c}:~y^2 + ay = x^3 + bx + c,
\label{eq:elliptic-curve}
\end{equation}
where $a, b, c \in\F_q .$  Let $S$ be the set of $x$-components of the affine points of ${\cal E}_{a,b,c}$ over $\F_q $, that is, 
\begin{equation}
S_{a,b,c} := \{\alpha\in \F_q| \exists \beta \in \F_q\text{ such that } \beta^2+a\beta=\alpha^2+b\alpha+c\}.
\label{eq:set-elliptic}
\end{equation}

For $q=2^m,$ any $\alpha\in S_{1,b,c}$ gives exactly two points with $x$-component $\alpha$, and we denote these two points corresponding to $\alpha$ by $P^{(1)}_\alpha$ and $P^{(2)}_\alpha.$ Then the set of all rational points of ${\cal E}_{1,b,c}$ over $\F_q$ is $\{P^{(1)}_\alpha|\alpha\in S_{1,b,c}\}\cup \{P^{(2)}_\alpha|\alpha\in S_{1,b,c}\}\cup \{P_\infty \}.$ The numbers of rational points of elliptic curves $\cal E$ over $\F_q$ are given in Table \ref{table:size-elliptic}.
\begin{table}
\caption{Numbers of rational points of elliptic curves}\label{table:size-elliptic}

\begin{center}
\begin{tabular}{c|c|c}
\text{Elliptic curve }${\cal E}_{1,b,c}$&$m$&$\# {\cal E}_{1,b,c}(\F_{2^m})$\\
\hline
&$m\text{ odd }$&$q+1-2\sqrt{q}$\\
$y^2 + y = x^3$&$m\equiv 0 \pmod 4$&$q+1-2\sqrt{q}$\\
&$m\equiv 2 \pmod 4$&$q+1+2\sqrt{q}$\\
\hline
\multirow{2}{8em} 
{$y^2+y=x^3+x$}& $m\equiv 1,7 \pmod 8$ & $q+1+2\sqrt{q}$\\
&$m\equiv 3,5 \pmod 8$ & $q+1-2\sqrt{q}$\\
\hline
\multirow{2}{8em}{$y^2 + y = x^3 + x+1$}&$m\equiv 1,7 \pmod 8$&$q+1+2\sqrt{q}$\\
&$m\equiv 3,5 \pmod 8$&$q+1-2\sqrt{q}$\\
\hline
$y ^2 + y = x ^3 + bx ( T r_1^m ( b) = 1 )$&$m\text{ even }$&$q+1$\\
\hline
\multirow{2}{13em}{$y ^2 + y = x ^3 + c ~( T r_1^ m ( c ) = 1 )$}&$m\equiv 0 \pmod 4$&$q+1+2\sqrt{q}$\\
&$m\equiv 2 \pmod 4$&$q+1-2\sqrt{q}$\\
\end{tabular}
\end{center}
\end{table}

\begin{lem}[Hilbert's Theorem 90]\label{lem:Hilbert90}Let $q=p^m.$ The equation $y^p-y=k$ has solutions over $\F_q$ if and only if $\text{Tr}_{\F_q/\F_p}(k)=0.$
\end{lem}
\begin{lem}\label{lem:trace0} Let $q_0=2^m, m\ge 2$ and $q=q_0^2.$ If $\alpha$ is an element in $\F_{q_0}$, then  $\text{Tr}_{\F_q/\F_2}(\alpha)=0,$ 
and $\text{Tr}_{\F_q/\F_2}(\alpha+\alpha^3)=0.$
\end{lem}
\begin{proof}
For any $\alpha \in \F_{q_0},$ we have 
$$
\begin{array}{ll}
\text{Tr}_{\F_q/\F_2}(\alpha)
&=\text{Tr}_{\F_{q_0}/\F_2}\left(\text{Tr}_{\F_q/\F_{q_0}}(\alpha)\right)\\
&=\text{Tr}_{\F_{q_0}/\F_2}\left(\alpha+\alpha^{q_0}\right)\\
&=\text{Tr}_{\F_{q_0}/\F_2}\left(\alpha \right)+\text{Tr}_{\F_{q_0}/\F_2}\left(\alpha^{q_0}\right)\\
&=0,
\end{array}
$$
where the first and second equality come from the properties of the trace function and the last one from the fact that $\alpha\in \F_{q_0}.$ Since $\F_{q_0}^*$ is a multiplicative group, the second part follows.

\end{proof}

\begin{prop} Let $q_0=2^m$ and $q=q_0^2.$ Then there exists a $[2q_0,q_0,d\ge q_0]$ self-dual code over $\F_q.$
\end{prop}
\begin{proof}Consider the elliptic curve defined by 
$${\cal E}_{1,1,0}:~y^2 + y = x^3+x.$$  From Lemma \ref{lem:trace0}, we get that $\F_{q_0}$ is a subset of $S_{1,1,0}$. Put $U=\F_{q_0}$ and $h(x)=\prod\limits_{\alpha\in U}(x-\alpha)$. Then the residue $\text{Res}_{P_\alpha}(\omega)=\frac{1}{h'(P_\alpha)}$ is a square for any $\alpha\in U$, and by Lemma \ref{lem:char1}, the constructed code $a\cdot C_{\cal L}(D,G)$ is self-dual, where $a_i^2=\text{Res}_{P_i}(\omega).$
\end{proof}

\begin{thm} \label{thm:elliptic2}
Let $q=2^m$ and $U=\{\alpha\in \F_q| \text{Tr}(\alpha^3+\alpha)=0\}.$ Then there exists a self-dual code over $\F_q$ with parameters $[2n,n,d\ge n]$  for $1\le n\le |U|.$
\end{thm}

\begin{proof} Let $U$ be defined as in the theorem.
Put 
$$h(x)=\prod\limits_{\alpha\in U}(x-\alpha).$$ 
Since any element in $\F_q$ ($q$ even) is a square in $\F_q$, we conclude that $h'(\alpha)$ is a nonzero square in $\F_q$ for any $\alpha\in U.$ 

Consider the elliptic curve defined by 
$${\cal E}_{1,1,0}:~y^2 + y = x^3+x.$$

From Lemma \ref{lem:Hilbert90}, we get that $U$ is a subset of $S_{1,1,0}$. Put $$D=\sum\limits_{\alpha \in U_0\subset U}\left(P_\alpha^{(1)}+P_\alpha^{(2)}\right)=P_1+\cdots+P_s,s=2|U_0|,G=\frac{s}{2}P_{\infty},\omega=\frac{dx}{h}.$$ Then the residue $\text{Res}_{P_\alpha}(\omega)=\frac{1}{h'(P_\alpha)}$ is a square for any $\alpha\in U_0$, and by Lemma \ref{lem:char1}, the constructed code $a\cdot C_{\cal L}(D,G)$ is self-dual, where $a_i^2=\text{Res}_{P_i}(\omega).$
\end{proof}
\begin{exam} The elliptic curve $${\cal E}_{1,1,0}:~y^2+y=x^3+x,$$
has rational points in the set
$
\{
P_\infty=(1:0:0),
    ( 1: 0 :1),
    ( 1: 1 :1),
    ( w^3: w^7 :1),
    ( w^3: w^9 :1),
    ( w^6: w^3 :1),
    ( w^6: w^{14} :1),
    ( w^{12}: w^6 :1),
    ( w^{12}: w^{13} :1),
    ( w^{10}: w :1),
    ( w^{10}: w^4 :1),
    ( w^{11}: w :1),
    ( w^{11}: w^4 :1),
    ( w^5: w^2 :1),
    ( w^5: w^8 :1),
    ( w^7: w^2 :1),
    ( w^7: w^8 :1),
    ( w^{13}: w^2 :1),
    ( w^{13}: w^8 :1)
\}
$. Put $D=P_1+\cdots+P_{18},G=9P_\infty$. The code $C_{\cal L}(D,G)$  is self-dual. The set $\{\frac{x^iy^j}{z^{i+j}}| (i,j)\in \{    
    ( 0, 0 ),
    ( 0, 1 ),
    ( 0, 2 ),
    ( 0, 3 ),
    ( 1, 0 ),
    ( 1, 1 ),
    ( 1, 2 ),
    ( 2, 0 ),
    ( 2, 1 )
\}
\}$ is a basis for $C_{\cal L}(D,G)$, and thus its generator matrix is given by

{\scriptsize
$$
{\cal G}=\left(
\begin{array}{cccccccccccccccccc}
1 1 1 1 1 1 1 1 1 1 1 1 1 1 1 1 1 1\\
0 1 w^7 w^9 w^3 w^{14} w^6 w^{13} w w^4 w w^4 w^2 w^8 w^2 w^8 w^2 w^8\\
0 1 w^{14} w^3 w^6 w^{13} w^{12} w^{11} w^2 w^8 w^2 w^8 w^4 w w^4 w w^4 w\\
0 1 w^6 w^{12} w^9 w^{12} w^3 w^9 w^3 w^{12} w^3 w^{12} w^6 w^9 w^6 w^9 w^6 w^9\\
1 1 w^3 w^3 w^6 w^6 w^{12} w^{12} w^{10} w^{10} w^{11} w^{11} w^5 w^5 w^7 w^7 w^{13} w^{13}\\
0 1 w^{10} w^{12} w^9 w^5 w^3 w^{10} w^{11} w^{14} w^{12} 1 w^7 w^{13} w^9 1 1 w^6\\
0 1 w^2 w^6 w^{12} w^4 w^9 w^8 w^{12} w^3 w^{13} w^4 w^9 w^6 w^{11} w^8 w^2 w^{14}\\
1 1 w^6 w^6 w^{12} w^{12} w^9 w^9 w^5 w^5 w^7 w^7 w^{10} w^{10} w^{14} w^{14} w^{11} w^{11}\\
0 1 w^{13} 1 1 w^{11} 1 w^7 w^6 w^9 w^8 w^{11} w^{12} w^3 w w^7 w^{13} w^4\\
\end{array}
\right).
$$
}

By Magma \cite{Mag}, the code with generator matrix $a\cdot {\cal G}$ is self-dual, and it has parameters $[18,9,9]$ over $\F_{16},$ where 
 $a=(    w^5,
    w^5,
    w^{12},
    w^{12},
    w^7,
    w^7,
    1,
    1,
    w^2,
    w^2,$ $
    w^{14},
    w^{14},
    w^4,
    w^4,
    w^9,
    w^9,
    w^{10},
    w^{10})
.$ 

This code is an almost MDS code. We also find almost MDS self-dual codes over $\F_{16}$ with parameters $[20,10,10]$, $[22,11,11]$, $[24,12,12].$
\end{exam}

\begin{cor} \label{cor:elliptic2}
Let $q=2^m$ and $U=\{\alpha\in \F_q| \text{Tr}(\alpha^3)=0\}.$ Then there exists a self-dual code over $\F_q$ with parameters $[2n,n,d \ge n]$  for $1\le n\le |U|.$
\end{cor}

\begin{thm} \label{thm:hyper-elliptic}
Let $q=2^m,m\ge 3$ and $U=\{\alpha\in \F_q| \text{Tr}(\alpha^5)=0\}.$ Then there exists a self-dual code over $\F_q$ with parameters $[2n,n,d\ge n-1]$  for $1\le n\le |U|.$
\end{thm}

\begin{proof} Let $U$ be defined as in the theorem.
Put 
$$h(x)=\prod\limits_{\alpha\in U}(x-\alpha).$$ 
Since any element in $\F_q$ ($q$ even) is a square in $\F_q$, we conclude that $h'(\alpha)$ is a nonzero square in $\F_q$ for any $\alpha\in U.$ 

Consider the hyper-elliptic curve defined by 
\begin{equation}
{\cal X}:~y^2 + y = x^5.
\label{eq:hyper-elliptic}
\end{equation}

From Lemma \ref{lem:Hilbert90}, we get that $U$ is a subset of the solution to (\ref{eq:hyper-elliptic}). Put $D=\sum\limits_{\alpha \in U_0\subset U}\left(P_\alpha^{(1)}+P_\alpha^{(2)}\right)=P_1+\cdots+P_s,s=2|U_0|,G=(\frac{s}{2}+1)P_{\infty}$ and $\omega=\frac{dx}{h}.$ Then the residue $\text{Res}_{P_\alpha}(\omega)=\frac{1}{h'(P_\alpha)}$ is a square for any $\alpha\in U_0$, and by Lemma \ref{lem:char1}, the constructed code $a\cdot C_{\cal L}(D,G)$ is a $[s,\frac{s}{2},d \ge \frac{s}{2}-1]$ self-dual code, where $a_i^2=\text{Res}_{P_i}(\omega).$ 

\end{proof}

\begin{exam}  The hyper-elliptic curve defined by
$$y^2+y=x^5,$$
has rational points in the set
$
\{
P_\infty=(1:0:0),
    ( w^{12}: 0 :1),
    ( w^{12}: 1 :1),
    ( w^3: 0 :1),
    ( w^3: 1 :1),
    ( w^6: 0 :1),
    ( w^6: 1 :1),
    ( w^9: 0 :1),
    ( w^9: 1 :1),
    ( w^{11}: w :1),
    ( w^{11}: w^4 :1),
    ( w^2: w :1),
    ( w^2: w^4 :1),
    ( w^5: w :1),
    ( w^5: w^4 :1),
    ( w^8: w :1),
    ( w^8: w^4 :1),
    ( w: w^2 :1),
    ( w: w^8 :1),
    ( w^{13}: w^2 :1),
    ( w^{13}: w^8 :1),
    ( w^4: w^2 :1),
    ( w^4: w^8 :1),
    ( w^7: w^2 :1),
    ( w^7: w^8 :1),
    ( w^{10}: w^2 :1),
    ( w^{10}: w^8 :1)
\}
$. Put $D=P_1+\cdots+P_{26},G=14P_\infty$. The set $\{\frac{x^iy^j}{z^{i+j}}| (i,j)\in \{    
    ( 0, 0 ),
    ( 0, 1 ),
    ( 0, 2 ),
    ( 1, 0 ),
    ( 1, 1 ),
    ( 1, 2 ),
    ( 2, 0 ),
    ( 2, 1 ),
    ( 2, 2 ),
    ( 3, 0 ),
    ( 3, 1 ),
    ( 4, 0 ),
    ( 4, 1 )
\}
\}$ is a basis for $C_{\cal L}(D,G)$, and thus its generator matrix $\cal G$ is given by

{\footnotesize

$$
{\cal G}=\left(
\begin{array}{cccccccccccccccccccccccccc}
1  1  1  1  1  1  1  1  1  1  1  1  1  1  1  1  1  1  1  1  1  1  1  1  1  1\\
0  1  0  1  0  1  0  1  w  w^4  w  w^4  w  w^4  w  w^4  w^2  w^8  w^2  w^8  w^2  w^8  w^2  w^8  w^2  w^8\\
0  1  0  1  0  1  0  1  w^2  w^8  w^2  w^8  w^2  w^8  w^2  w^8  w^4  w  w^4  w  w^4  w  w^4  w  w^4  w\\
w^{12}  w^{12}  w^3  w^3  w^6  w^6  w^9  w^9  w^{11}  w^{11}  w^2  w^2  w^5  w^5  w^8  w^8  w  w  w^{13}  w^{13}  w^4  w^4  w^7  w^7  w^{10}  w^{10}\\
0  w^{12}  0  w^3  0  w^6  0  w^9  w^{12}  1  w^3  w^6  w^6  w^9  w^9  w^{12}  w^3  w^9  1  w^6  w^6  w^{12}  w^9  1  w^{12}  w^3\\
0  w^{12}  0  w^3  0  w^6  0  w^9  w^{13}  w^4  w^4  w^{10}  w^7  w^{13}  w^{10}  w  w^5  w^2  w^2  w^{14}  w^8  w^5  w^{11}  w^8  w^{14}  w^{11}\\
w^9  w^9  w^6  w^6  w^{12}  w^{12}  w^3  w^3  w^7  w^7  w^4  w^4  w^{10}  w^{10}  w  w  w^2  w^2  w^{11}  w^{11}  w^8  w^8  w^{14}  w^{14}  w^5  w^5\\
0  w^9  0  w^6  0  w^{12}  0  w^3  w^8  w^{11}  w^5  w^8  w^{11}  w^{14}  w^2  w^5  w^4  w^{10}  w^{13}  w^4  w^{10}  w  w  w^7  w^7  w^{13}\\
0  w^9  0  w^6  0  w^{12}  0  w^3  w^9  1  w^6  w^{12}  w^{12}  w^3  w^3  w^9  w^6  w^3  1  w^{12}  w^{12}  w^9  w^3  1  w^9  w^6\\
w^6  w^6  w^9  w^9  w^3  w^3  w^{12}  w^{12}  w^3  w^3  w^6  w^6  1  1  w^9  w^9  w^3  w^3  w^9  w^9  w^{12}  w^{12}  w^6  w^6  1  1\\
0  w^6  0  w^9  0  w^3  0  w^{12}  w^4  w^7  w^7  w^{10}  w  w^4  w^{10}  w^{13}  w^5  w^{11}  w^{11}  w^2  w^{14}  w^5  w^8  w^{14}  w^2  w^8\\
w^3  w^3  w^{12}  w^{12}  w^9  w^9  w^6  w^6  w^{14}  w^{14}  w^8  w^8  w^5  w^5  w^2  w^2  w^4  w^4  w^7  w^7  w  w  w^{13}  w^{13}  w^{10}  w^{10}\\
0  w^3  0  w^{12}  0  w^9  0  w^6  1  w^3  w^9  w^{12}  w^6  w^9  w^3  w^6  w^6  w^{12}  w^9  1  w^3  w^9  1  w^6  w^{12}  w^3\\
1  1  1  1  1  1  1  1  w^{10}  w^{10}  w^{10}  w^{10}  w^{10}  w^{10}  w^{10}  w^{10}  w^5  w^5  w^5  w^5  w^5  w^5  w^5  w^5  w^5  w^5\\
w^{12}  w^{12}  w^3  w^3  w^6  w^6  w^9  w^9  w^6  w^6  w^{12}  w^{12}  1  1  w^3  w^3  w^6  w^6  w^3  w^3  w^9  w^9  w^{12}  w^{12}  1  1\\
w^9  w^9  w^6  w^6  w^{12}  w^{12}  w^3  w^3  w^2  w^2  w^{14}  w^{14}  w^5  w^5  w^{11}  w^{11}  w^7  w^7  w  w  w^{13}  w^{13}  w^4  w^4  w^{10}  w^{10}
\end{array}
\right).
$$
}

By Magma \cite{Mag}, the code with generator matrix $a\cdot {\cal G}$ is self-dual, and it has parameters $[26,13,12]$ over $\F_{16},$ where 
$a=(          
w^{14},
    w^{14},
    w,
    w,
    w^6,
    w^6,
    w^{10},
    w^{10},
    w^9,$ $
    w^9,
    w^4,
    w^4,
    w^6,
    w^6,
    w^8,
    w^8,
    w^6,
    w^6,
    w^3,
    w^3,
    w^7,
    w^7,
    w,
    w,
    w^{13},
    w^{13}
)
.$ 
We also find self-dual codes over $\F_{16}$ with parameters $[28,14,13]$, $[30,15,14]$, $[32,16,15].$
\end{exam}

Next, we will consider hyper-elliptic curves over $\F_q,$ $q=p^m$ with $p$ an odd prime.

\begin{thm}\label{thm:hyper-elliptic}
Let $q=p^m$ and $t$ be a positive odd integer such that $\text{gcd}(t,q-1)=1$. If $\eta(n)=1$ and $4n|(q-1)$, then there exists a self-dual code with parameters $[2n,n,d\ge n+\frac{t-3}{2}].$
\end{thm}

\begin{proof} Denote $C_j=\{j\times i \pmod {q-1}| i=0,1,\hdots\}$. For $\theta$ a primitive element of $\F_q$, let $U_n=\{\theta^{i }| i\in C_{\frac{q-1}{n}}\},$ and label the elements of $U_n$ as $\alpha_1,\hdots,\alpha_n$. Under the condition $4n|(q-1)$, the set $U_n$ is a multiplicative subgroup of $\F_q^*$ of order $n.$
Put 
$$h(x)=\prod\limits_{\alpha\in U_n}(x-\alpha).$$ 
Clearly,  all the roots of $h(x)$ are simple, and the derivative $h'(x)=nx^{n-1}$, and thus for any $\alpha\in U_n$, we have that $h'(\alpha)$ is a square.
Consider the elliptic curve defined by 
$${\cal X}:~y^2 = x^t.$$
Since $\text{gcd}(t,q-1)=1$, the set $\{x^t| x \in \F_q\}$ is in bijection with $\F_q.$ For any $\alpha\in U_n$, there are two places, say $P_\alpha^{(1)}$ and $P_\alpha^{(2)}$, arising from $x$-component $\alpha.$ Put $D=\sum\limits_{\alpha \in U_n}P_\alpha^{(1)}+P_\alpha^{(2)}=P_1+\cdots+P_s,s=2n,G=nP_{\infty}$ and $\omega=\frac{dx}{h}.$ With the choice of $\alpha_i\in U_n$ and $\beta_i^2=\alpha_i^t$, the residue $\text{Res}_{P_{\alpha_i}}(\omega)=\frac{1}{\beta_ih'(\alpha_i)}$ is a square for any $\alpha_i\in U_n$, and by Lemma \ref{lem:char1}, the constructed code $a\cdot C_{\cal L}(D,G)$ is self-dual, where $a_i^2=\text{Res}_{P_i}(\omega).$
\end{proof}
\begin{cor}Let $q=p^m$. Then we have the following:
\begin{enumerate} 
\item if $\text{gcd}(3,q-1)=1$, $\eta(n)=1$ and $4n|(q-1)$, then there exists a self-dual code with parameters $[2n,n,d\ge n];$
\item  If $\text{gcd}(5,q-1)=1$, $\eta(n)=1$ and $4n|(q-1)$, then there exists a self-dual code with parameters $[2n,n,d\ge n-1].$
\end{enumerate}
\end{cor}


\begin{exam} The hyper-elliptic curve over $\F_{25}$ defined by $$y^2=x^5,$$
has rational points in the set
$
\{
P_\infty=(1:0:0),
    ( 1: 1 :1),
    ( 1: 4 :1),
    ( w^{20}: w^2 :1),
    ( w^{20}: w^{14} :1),
    ( w^{16}: w^4 :1),
    ( w^{16}: w^{16} :1),
    ( 4: 2 :1),
    ( 4: 3 :1),
    ( w^8: w^8 :1),
    ( w^8: w^{20} :1),
    ( w^4: w^{10} :1),
    ( w^4: w^{22} )
\}
.$ Put $D=P_1+\cdots+P_{12}$ and $G=7P_\infty$. The set $\{\frac{x^iy^j}{z^{i+j}}| (i,j)\in \{    
    ( 0, 0 ),
    ( 0, 1 ),
    ( 1, 0 ),
    ( 1, 1 ),
    ( 2, 0 ),
    ( 3, 0 )
\}
\}$ is a basis for $C_{\cal L}(D,G)$, and thus its generator matrix is given by

$$
{\cal G}=\left(
\begin{array}{cccccccccccc}
   1&1&1&1&1&1&1&1&1&1&1&1\\
   1&4&  w^{2} &w^{14}&  w^{4} &w^{16}&2&3&  w^{8} &w^{20} &w^{10} &w^{22}\\
   1&1 &w^{20} &w^{20} &w^{16} &w^{16}&4&4&  w^{8}&  w^{8}&  w^{4}&  w^{4}\\
   1&4 &w^{22} &w^{10} &w^{20}&  w^{8}&3&2& w^{16}&  w^{4} &w^{14}&  w^{2}\\
   1&1 &w^{16} &w^{16}&  w^{8}&  w^{8}&1&1& w^{16} &w^{16}&  w^{8}&  w^{8}\\
   1&1&4&4&1&1&4&4&1&1&4&4\\
\end{array}
\right).
$$

By Magma \cite{Mag}, the code with generator matrix $a\cdot {\cal G}$ is self-dual, and it has parameters $[12,6,5]$ over $\F_{25},$ where 
$a=(
    1,
    3,
    w^{21},
    w^{15},
    3,
    1,
    w^{15},
    w^{21},
    1,
    3,
    w^{21},
    w^{15}
)
.$ 
\end{exam}
By considering curves in higher genus, we can release the gcd condition in Theorem \ref{thm:hyper-elliptic}.
\begin{thm}
Let $q=p^m$ with $p$ an odd prime.
\begin{enumerate}
 \item If $n$ is odd, $\eta(n)=1$ and $4n|(q-1)$, then there exists a self-dual code with parameters $[2n,n,d\ge \frac{n}{2}+2].$
\item If $n$ is even, $\eta(n)=1$ and $2n|(q-1)$, then there exists a self-dual code with parameters $[2n,n,d\ge \frac{n}{2}+2].$
\end{enumerate}
\end{thm}
\begin{proof} Assume that $n$ is odd. Let $U_n$ and $h(x)$ be defined as in Theorem \ref{thm:hyper-elliptic}. Consider an algebraic curve given by 
$${\cal X}:~y^2 = x^n.$$
Take $\omega =\frac{dx}{h}$ and $G=\frac{3n-4}{2}P_\infty.$ Then by Lemma \ref{lem:char1}, the code $a\cdot C_{\cal L}(D,G)$, where $a_i=\text{Res}_{P_i}(\omega)$, is self-dual with parameters $[2n,n,\frac{n}{2}+2],$ and this proves point 1). 

For point 2), we put $U_n=\{a^2| a \in \F_q, a^n=1\}.$ The rest follows from the same reasoning as the first part.
\end{proof}
\subsection{Self-dual codes from other curves}
In this subsection, we will construct self-dual codes over $\F_q$ from algebraic curves of high genus.

Let $q_0=p^m,q=q_0^2$ and ${\cal X}$ be the Hermitian
curve over $\F_{q}$ defined by
$${\cal X:~}y^{q_0} + y = x^{{q_0}+1}.$$ 
The Hermitian curve ${\cal X}$ has genus $g = \frac{q_0(q_0 -1)}{2},$ and for any $\alpha\in \F_{q},$ $x-\alpha$  has $q$ zeros of degree one in ${\cal X}$. All rational points  of the curve $\cal X$ different from the point at infinity are obtained in this way. Self-orthogonal AG codes from Hermitian curves were already considered in \cite{Stich88}. In what follows, we embed those codes into the self-dual ones and provide the parameters of the latter codes. We also construct new families of self-dual codes from this curve.

\begin{thm}\label{thm:hermitian-curve}
Let $p$ be an odd prime, $q_0=p^m$, $q=q_0^2,g=\frac{q_0(q_0-1)}{2}$. Put $d_0=\frac{s}{2}+1-g, s'=s+1,d_0'=\frac{s'}{2}-g.$
\begin{enumerate}
\item If $p|n,(n-1)|(q-1)$, then there exists a $q$-ary self-dual code with parameters $[s,\frac{s}{2},d\ge d_0]$ (resp. $[s',\frac{s'}{2},d\ge d_0']$), where $s=q_0n$ with $n$ even (resp. $n$ odd).
\item If $r|m$ , then there exists a $q$-ary self-dual code with parameters $[s',\frac{s'}{2},d\ge d_0']$, where $s=q_0p^r.$
\item If $(n-1)|(q-1)$, then there exists a $q$-ary self-dual code with parameters $[s',\frac{s'}{2},d\ge d_0']$, where $s=q_0(2n-1).$
\item If $n|\frac{q-1}{2}$, then there exists a $q$-ary self-dual code with parameters $[s,\frac{s}{2},d\ge d_0]$ (resp. $[s',\frac{s'}{2},d\ge d_0']$), where $s=q_0n$ with $n$ even (resp. $n$ odd).
\item If $1\le r< m$, then there exists a $q$-ary self-dual code with parameters $[s',\frac{s'}{2},d\ge d_0']$, where $s=q_0(2p^r-1).$
\item If $n=q_0-1,n\equiv 0 \pmod 4$, then there exists a $q$-ary self-dual code with parameters $[s',\frac{s'}{2},d\ge d_0']$, where $s=q_0\left(n(t+1)+1\right)$, for $t$ odd, $0\le t\le \frac{n}{2}+1.$
\item If $n=q_0-1,n\equiv 2 \pmod 4$, then there exists a $q$-ary self-dual code with parameters $[s',\frac{s'}{2},d\ge d_0']$, where $s=q_0\left(n(t+1)+1\right)$, $0\le t\le \frac{n}{2}.$

\item If $1\le r< m$ and $\frac{n(p^r+1)}{2(p^r-1)}$ is odd, then there exists a $q$-ary self-dual code with parameters $[s,\frac{s}{2},d\ge d_0]$ (resp. $[s',\frac{s'}{2},d\ge d_0']$), where $s=q_0(t+1)n$ (resp. $s=q_0\left((t+1)n+1\right)+1$), $n=\frac{q-1}{p^r+1}$, for $t$ odd , $1\le t\le p^r.$
\item If $1\le r< m$ and $\frac{n(p^r+1)}{2(p^r-1)}$ is even, then there exists a $q$-ary self-dual code with parameters $[s,\frac{s}{2},d\ge d_0]$  (resp. $[s',\frac{s'}{2},d\ge d_0']$), where $s=q_0(t+1)n$ (resp. $s=q_0\left((t+1)n+1\right)+1$), $n=\frac{q-1}{p^r+1}$, for $1\le t\le p^r.$
\item If $1\le r< m$, then there exists a $q$-ary self-dual code with parameters $[s,\frac{s}{2},d\ge d_0]$  (resp. $[s',\frac{s'}{2},d\ge d_0']$), where  $s=q_0(t+1)n$ (resp. $s=q_0\left((t+1)n+1\right)+1$), $n=\frac{q-1}{p^r-1},r|\frac{m}{2}$, for $1\le t\le p^r-2.$
\item  If $t$ is even such that $1\le t\le q_0$, then there exists a $q$-ary self-dual code with parameters $[s,\frac{s}{2},d\ge d_0]$, where $s=q_0(q_0t).$
\item  If $t$ is odd such that $1\le t\le q_0$, then there exists a $q$-ary self-dual code with parameters $[s',\frac{s'}{2},d\ge d_0']$, where $s=q_0(q_0t).$

\item If $r=p^k,k|m,0\le \ell<m/k,1\le t\le (r-1)/2$, then there exists a $q$-ary self-dual code with parameters $[s,\frac{s}{2},d\ge d_0]$, where $s=q_0(2tr^\ell).$
\item If $0\le \ell < 2m$, then there exists a $q$-ary self-dual code with parameters $[s,\frac{s}{2},d\ge d_0]$, where $ s=q_0(2p^\ell).$
\item If $r=p^k,k|m,0\le \ell<m/k,0\le t\le (r-1)/2$ or $(\ell, t)=(m/k,0)$, then there exists a $q$-ary self-dual code with parameters $[s',\frac{s'}{2},d\ge d_0']$, where $s=q_0(2t+1)r^\ell.$
\item If $0\le \ell < 2m$, then there exists a $q$-ary self-dual code with parameters $[s',\frac{s'}{2},d\ge d_0']$, where $ s=q_0p^\ell.$

\end{enumerate}

\end{thm}
\begin{proof}It should be noted that each $x$-component $\alpha\in \F_q$ gives $q$ places of degree one. Let $U$ be a subset of $\{\alpha\in \F_q| \beta^{q_0}+\beta=\alpha^{q_0+1}\}$ such that $q|U|=s.$ Put $$h(x)=\prod\limits_{\alpha\in U}(x-\alpha)\text{ and }\omega =\frac{dx}{h}.$$ For each case, it is enough to prove that the residue $\text{Res}_{P_\alpha}(\omega)$ of $\omega$ at place $P_\alpha$ is a nonzero square for any $\alpha\in U$, that is, $h'(\alpha)$ is a nonzero square in $\F_q$. Take $U$ as follows.
\begin{itemize}
\item for 1), $U=\{\alpha\in \F_q| \alpha^n=\alpha\},$
\item for 2), $U=\{\alpha \in \F_q|\alpha^{p^r}=\alpha\},$
\item for 3), $U=U_{n-1}\cup \alpha_1U_{n-1}$, where $U_{n-1}=\{\alpha \in \F_q|\alpha^{n-1}=1\}$ and $\alpha_1\in \F_q\backslash U_{n-1}$ such that $1-\alpha_1^{n-1}$ is a square,
\item for 4), $U=\{\alpha \in \F_q|\alpha^n=1\},$
\item for 5), $U=U_{n}\cup \alpha_1U_{n}\cup \{0\}$, where $n=p^r-1, U_{n}=\{\alpha \in \F_q|\alpha^{n}=1\}$ and $\alpha_1\in \F_q\backslash U_{n}$ such that $1-\alpha^n$ is a square,
\item for  6)--10), take $U=\{0\}\cup V\text { or }U=V, \text{ where }V=U_{n}\cup \alpha_1U_{n}\cup \cdots \cup \alpha_tU_n$, $U_{n}=\{\alpha \in \F_q|\alpha^{n}=1\}$ and $\alpha_1,\hdots,\alpha_t\in \F_q\backslash U_{n}$ as in Theorem \ref{thm:new-multicoset},
\item for 10)--12), label the elements of $\F_{q_0}$ as $a_1,\hdots ,a_{q_0}. $ For some fixed element $\beta\in \F_q\backslash \F_{q_0}, $ take $U=\{a_k\beta+a_j|1\le k, j\le q_0\},$
\item for 13)--16),  label the element of $\F_r$ as $a_0,\hdots,a_{r-1}$, take $H$ as an $\F_r$-subspace and set $H_i=H+a_i \beta$ for some fixed element 
$\beta \in \F_q\backslash \F_r.$ Put $U=H_0\cup \cdots \cup H_{2t-1}$ or $U=H_0\cup \cdots \cup H_{2t}$. 
\end{itemize}

For 1)--5), it can be easily checked that $h'(\alpha)$ is a square for any $\alpha\in U.$

For 6)--10), it has been already checked, in Theorem \ref{thm:new-multicoset}, that $h'(\alpha)$ is a square for any $\alpha\in U.$

For 11)--12), it was proved in \cite[Theorem 2]{Yan} that $h'(\alpha)$ is a square for any $\alpha \in U.$

For 13)--16),  it was proved in \cite[Theorem 4]{FangFu} that $h'(\alpha)$ is a square for any $\alpha \in U.$
\end{proof}
\begin{exam} The Hermitian curve defined over $\F_{9}$ has all rational points in the set
$
\{
P_\infty=(1:0:0),
    ( 0: 0 :1),
    ( 0: w^2 :1),
    ( 0: w^6 :1),
    ( 1: w :1),
    ( 1: w^3 :1),
    ( 1: 2 :1),
    ( w^2: w :1),
    ( w^2: w^3 :1),
    ( w^2: 2 :1),
    ( 2: w :1),
    ( 2: w^3 :1),
    ( 2: 2 :1),
    ( w^6: w :1),
    ( w^6: w^3 :1),
    ( w^6: 2 :1),
    ( w: 1 :1),
    ( w: w^5 :1),
    ( w: w^7 :1),
    ( w^3: 1 :1),
    ( w^3: w^5 :1),
    ( w^3: w^7 :1),
    ( w^5: 1 :1),
    ( w^5: w^5 :1),
    ( w^5: w^7 :1),
    ( w^7: 1 :1),
    ( w^7: w^5 :1),
    ( w^7: w^7 :1)
\}.$
Put $D=P_1+\cdots+P_{27},G=15P_\infty$. The code $C_{\cal L}(D,G)$  has parameters $[27,13,12]$. The set
$\{\frac{x^iy^j}{z^{i+j}}| (i,j)\in \{    
     ( 0, 0 ),
    ( 0, 1 ),
    ( 0, 2 ),
    ( 0, 3 ),
    ( 1, 0 ),
    ( 1, 1 ),
    ( 1, 2 ),
    ( 1, 3 ),
    ( 2, 0 ),
    ( 2, 1 ),$ $
    ( 2, 2 ),
    ( 3, 0 ),
   ( 3, 1 )
\}\}$
 is a basis for the code $C_{\cal L}(D,G)$, and thus its generator matrix is given by

{\tiny

$$
{\cal G}=\left(
\begin{array}{ccccccccccccccccccccccccccc}
1 1 1 1 1 1 1 1 1 1 1 1 1 1 1 1 1 1 1 1 1 1 1 1 1 1 1\\
0  w^2  w^6 w  w^3  2 w  w^3  2 w  w^3  2 w  w^3  2 1  w^5  w^7 1  w^5  w^7 1  w^5  w^7 1  w^5
  w^7\\
0  2  2  w^2  w^6 1  w^2  w^6 1  w^2  w^6 1  w^2  w^6 1 1  w^2  w^6 1  w^2  w^6 1  w^2  w^6 1
  w^2  w^6\\
0  w^6  w^2  w^3  w 2  w^3  w 2  w^3  w 2  w^3  w 2 1  w^7  w^5 1  w^7  w^5 1  w^7  w^5 1  w^7  w^5\\
0 0 0 1 1 1  w^2  w^2  w^2  2  2  2  w^6  w^6  w^6  w w w  w^3  w^3  w^3  w^5  w^5  w^5  w^7  w^7
  w^7\\
0 0 0 w  w^3  2  w^3  w^5  w^6  w^5  w^7 1  w^7 w  w^2 w  w^6 1  w^3 1  w^2  w^5  w^2  2  w^7  2
  w^6\\
0 0 0  w^2  w^6 1  2 1  w^2  w^6  w^2  2 1  2  w^6 w  w^3  w^7  w^3  w^5 w  w^5  w^7  w^3  w^7 w
  w^5\\
0 0 0  w^3  w 2  w^5  w^3  w^6  w^7  w^5 1 w  w^7  w^2  w 1  w^6  w^3  w^2 1  w^5  2  w^2  w^7
  w^6  2\\
0 0 0 1 1 1  2  2  2 1 1 1  2  2  2  w^2  w^2  w^2  w^6  w^6  w^6  w^2  w^2  w^2  w^6  w^6  w^6\\
0 0 0 w  w^3  2  w^5  w^7 1 w  w^3  2  w^5  w^7 1  w^2  w^7 w  w^6  w^3  w^5  w^2  w^7 w  w^6
  w^3  w^5\\
0 0 0  w^2  w^6 1  w^6  w^2  2  w^2  w^6 1  w^6  w^2  2  w^2  2 1  w^6 1  2  w^2  2 1  w^6 1  2\\
0 0 0 1 1 1  w^6  w^6  w^6  2  2  2  w^2  w^2  w^2  w^3  w^3  w^3  w w w  w^7  w^7  w^7  w^5  w^5
  w^5\\
0 0 0 w  w^3  2  w^7 w  w^2  w^5  w^7 1  w^3  w^5  w^6  w^3 1  w^2 w  w^6 1  w^7  2  w^6  w^5
  w^2  2\\
\end{array}
\right).
$$
}
Take ${\cal G'}=\left(\begin{array}{c} {\cal G}~0\\g_{14} \end{array}\right),$ where 
{ $g_{14}=(0, 1, 1,  2,  2, 1,  2,  2, 1,  2,  2, 1,  2,  2, 1, 1,  2,  2,$ $ 1,  2,  2, 1,  2,  2, 1,  2,  2, 1).$} 
By Magma \cite{Mag}, the code with generator matrix $a\cdot{\cal G'}$ is an optimal $[28,14,12]$ self-dual code with new parameters, where $a_i= w^6,1\le i\le 28.$ Other parameters from different constructions in Theorem \ref{thm:hermitian-curve} are given in Table \ref{table:hermitian-code}.
\end{exam}

\begin{table}[h]

\caption{Self-dual codes of length $s$  from Hermitian curves defined over $\F_9$}\label{table:hermitian-code}
{\scriptsize
$$
\begin{array}{cccc|ccc}
\text{Theorem }\ref{thm:hermitian-curve}&\text{Length } s&\text{Distance }&\text{Lower bound}&\text{Extended length}&\text{Distance}&\text{Lower bound}\\
\hline
\multirow{2}{3em}{1)}&3.3&3&3&10&3&2\\
&3.9&12&12&28&12&11\\
\hline
\multirow{1}{3em}{2)}&3.3&3&3&10&3&2\\
\hline
\multirow{1}{3em}{3)}&3(2.3-1)&6&6&16&6&5\\
\hline
\multirow{2}{3em}{4)}&3.2&3&1&-&-&-\\
&3.4&4&4&-&-&-\\
\hline
\multirow{1}{3em}{5)}&3(1+1)2&4&4&-&-&-\\
\hline
\multirow{1}{3em}{6)}&3\left((1+1)2+1\right)&6&6&16&6&5\\
\hline
\multirow{1}{3em}{11)}&3\left(3.2\right)&7&7&-&-&-\\
\end{array}
$$
}
\end{table}

\begin{thm}\label{thm:hermitian-2} Let $q_0=p^m, q=q_0^2$ be an odd prime power, $g=\frac{(q_0-1)^2}{4}$ and $s={q_0}\frac{q_0^2+1}{2}$. Then,
\begin{enumerate}
\item  there exists a $[s,\frac{s}{2},d\ge \frac{s}{2}-g+1]$ self-dual code if $s$ is even, and
\item there exists a $[s+1,\frac{s+1}{2},d\ge \frac{s+1}{2}-g]$ self-dual code if $s$ is odd.
\end{enumerate}
\end{thm}
\begin{proof}
Consider an algebraic curve defined by
$${\cal X}: y^{q_0}+y=x^{\frac{{q_0}+1}{2}}.$$
The curve has genus $g=\frac{(q_0-1)^2}{4}$.
Put $$U=\{\alpha\in \F_q|\exists \beta\in \F_q\text{ such that }\beta^{q_0}+\beta=\alpha^{\frac{q_0+1}{2}}\}.$$ 
The set $U$ is the set of $x$-component solutions to the Hermitian curve whose elements are squares in $\F_q.$ There are $\frac{q_0^2+1}{2}$ square elements in $\F_q$, and this gives rise to ${q_0}\frac{q_0^2+1}{2}$ rational places. Write
$$h(x)=\prod\limits_{\alpha\in U}(x-\alpha)\text{ and } \omega=\frac{dx}{h}.$$
Then $h(x)=x^n-x$, where $n=\frac{q_0^2+1}{2},$ and thus $h'(x)=nx^{n-1}-1$. Since $q$ is a square, we have that $h'(\alpha)=n-1$ is a square for any $\alpha\in U \backslash \{0\}$.
Put $D=\sum\limits_{\alpha\in U}\left(P_\alpha^{(1)}+\cdots+P_\alpha^{(q_0)}\right)=P_1+\cdots+P_s,s=q_0\frac{q_0^2+1}{2}.$ Set 
$$G=
\begin{cases}
(g-1+\frac{s}{2})P_{\infty}\text{ if } s\text{ is even,}\\
(g-1+\frac{s-1}{2})P_{\infty}\text{ if } s\text{ is odd.}\\
\end{cases}
$$ Then the residue $\text{Res}_{P_{\alpha}}(\omega)=\frac{1}{h'(P_{\alpha})}$ is a square for any $\alpha\in U$, and by Lemma \ref{lem:char1}, the constructed code $a\cdot C_{\cal L}(D,G)$ is self-orthogonal, where $a_i^2=\text{Res}_{P_i}(\omega).$ If $s$ is even, then point 1) follows, otherwise the self-orthogonal code can be embedded into a self-dual code using Lemma \ref{lem:embedding2-odd-g}, and thus point 2) follows.
\end{proof}

\begin{exam} There exist self-dual codes with parameters 
$[ 16, 8, \ge 7 ]_{3^2
}$, $[ 66, 33, \ge 29 ]_{
5^2
}$, $[ 176, 88, \ge 79 ]_{
7^2
}$, $[ 370, 185, \ge 169 ]_{
9^2
}$, $[ 672, 336, \ge 311 ]_{
11^2
}$, $[ 1106, 553, \ge 517 ]_{
13^2
}$, $[ 2466, 1233, \ge 1169 ]_{
17^2
}$, $[ 3440, 1720, \ge 1639 ]_{
19^2
}$, $[ 7826, 3913, \ge 3769 ]_{
25^2}$.
We now calculate the exact distance of the self-dual code over $\F_{3^2}$. The algebraic curve over $\F_{9}$ defined by

$$y^3+y=x^2$$

has all rational points in the set
$
\{
P_\infty=(1:0:0),
    ( 0: 0 :1),
    ( 0: w^2 :1),
    ( 0: w^6 :1),
    ( 1: w :1),
    ( 1: w^3 :1),
    ( 1: 2 :1),
    ( 2: w :1),
    ( 2: w^3 :1),
    ( 2: 2 :1),
    ( w^2: 1 :1),
    ( w^2: w^5 :1),
    ( w^2: w^7 :1),
    ( w^6: 1 :1),
    ( w^6: w^5 :1),
    ( w^6: w^7:1)
\}.$
Put $D=P_1+\cdots+P_{15},G=7P_\infty$. The code $C_{\cal L}(D,G)$  has parameters $[15,7,8]$. 

The set
$\{\frac{x^iy^j}{z^{i+j}}| (i,j)\in \{    
     ( 0, 0 ),
    ( 0, 1 ),
    ( 0, 2 ),
    ( 0, 3 ),
    ( 1, 0 ),
    ( 1, 1 ),
    ( 1, 2 )
\}\}$
 is a basis for the code $C_{\cal L}(D,G)$, and thus its generator matrix is given by

$$
{\cal G}=\left(
\begin{array}{ccccccccccccccc}
  1&1&1&1&1&1&1&1&1&1&1&1&1&1&1\\
  0& w^2& w^6&w& w^3&2&w& w^3&2&1& w^5& w^7&1& w^5& w^7\\
  0&2&2& w^2& w^6&1& w^2& w^6&1&1& w^2& w^6&1& w^2& w^6\\
  0& w^6& w^2& w^3&w&2& w^3&w&2&1& w^7& w^5&1& w^7& w^5\\
  0&0&0&1&1&1&2&2&2& w^2& w^2& w^2& w^6& w^6& w^6\\
  0&0&0&w& w^3&2& w^5& w^7&1& w^2& w^7&w& w^6& w^3& w^5\\
  0&0&0& w^2& w^6&1& w^6& w^2&2& w^2&2&1& w^6&1&2\\
\end{array}
\right).
$$


Take $${\cal G'}=\left(\begin{array}{c} a\cdot {\cal G}~0\\g_{8} \end{array}\right),$$ where 
$a=( w^6,
    w^6,
    w^6,
    1,
    1,
    1,
    1,
    1,
    1,
    1,
    1,
    1,
    1,
    1,
    1
)$,
 $g_{8}=
(    0,   w,   0,   0,   0,   0,   0,   0, w^7, $ $  0, w^5, w^2, w^5, w^7, w^7,   2).$
By Magma \cite{Mag}, the code with generator matrix ${\cal G'}$ is a $[16,8,8]$ self-dual code, which is optimal and has new parameters.

\end{exam}

\begin{thm} Let $q_0=p^m, q=q_0^2$ be an odd prime power, $g=\frac{(q_0-1)^2}{4}$ and $s={q_0}\frac{q_0^2-1}{2}$. Then,
\begin{enumerate}
\item  there exists a $[s,\frac{s}{2},d\ge \frac{s}{2}-g+1]$ self-dual code if $s$ is even, and
\item there exists a $[s+1,\frac{s+1}{2},d\ge \frac{s+1}{2}-g]$ self-dual code if $s$ is odd.
\end{enumerate}
\end{thm}
\begin{proof} Consider the same setting as the proof of Theorem \ref{thm:hermitian-2}. Take $U'=U\backslash \{0\},$ and write 
$$h(x)=\prod\limits_{\alpha\in U'}(x-\alpha)\text{ and } \omega=\frac{dx}{h}.$$
The rest follows with the same reasoning as that in Theorem \ref{thm:hermitian-2}.
\end{proof}

\begin{exam}There exist self-dual codes with parameters 
$[ 12, 6, 6 ]_{
3^2
}$, $[ 60, 30, \ge 27 ]_{
5^2
}$, $[ 168, 84, \ge 76 ]_{
7^2
}$, $[ 360, 180, \ge 165 ]_{
9^2
}$, $[ 660, 330, \ge 306 ]_{
11^2
}$, $[ 1092, 546, \ge 511 ]_{
13^2
}$, $[ 2448, 1224, \ge 1161 ]_{
17^2
}$, $[ 3420, 1710, \ge 1630 ]_{
19^2
}$, $[ 7800, 3900, \ge 3757 ]_{
25^2}.$
\end{exam}

We update parameters of MDS self-dual codes from the previous constructions in Table \ref{table:2}.

\begin{table}[th]
\caption{MDS self-dual codes of length $n$ over $\F_q$, $-$: no self-dual code exists with such a pair $(n,q)$, $+$: known parameters, $?$: unknown parameters, $^*$: new parameters}
{\scriptsize
$$
\begin{array}{c|c|c|c|c|c|c|c|c|c|c|c|c|c|c|c|c|c|c|c|c}
n/q&11&13&16&17&19&23&25&27&29&31&32&37&41&43&47&49&53&61&73&81\\
\hline
2&-&+&+&+&-&-&+&-&+&-&+&+&+&-&-&+&+&+&+&+\\
\hline
4&+&+&+&+&+&+&+&+&+&+&+&+&+&+&+&+&+&+&+&+\\
\hline
6&-&+&+&+&-&-&+&-&+&-&+&+&+&-&-&+&+&+&+&+\\
\hline
8&+&+&+&+&+&+&+&+&+&+&+&+&+&+&+&+&+&+&+&+\\
\hline
10&-&+&+&+&-&-&+&-&+&-&+&+&+&-&-&+&+&+&+&+\\
\hline
12&+&6&+&+&+&+&+&+&+&+&+&+&+&+&+&+&+&+&+&+\\
\hline
14&-&+&+&?&-&-&+&-&+&-&+&+&+&-&-&+&+&+&+&+\\
\hline
16&&&+&?&?&+&+&+&+&+&+&+&+&+&+&+&?&+&?&+\\
\hline
18&&&&+&-&-&+&-&+&-&+&+&+&-&-&+&?&?&+&+\\
\hline
20&&&&&+&?&+&?&?&+&+&?&+&+&+&+&?&+&+&+\\
\hline
22&-&&&&-&-&?&-&?&-&+&?&+&?&?&+&?&+&?&+\\
\hline
24&&&&&&+&?&?&?&+&+&?&?&+&+&+&?&?&+&*\\
\hline
26&-&&&&-&-&+&-&?&-&+&*&?&-&-&+&+&*&+&+\\
\hline
28&&&&&&&&+&?&?&+&?&?&?&?&+&?&?&?&+\\
\hline
30&-&&&&-&-&&&+&&+&?&?&-&-&?&?&+&?&+\\
\hline
32&&&&&&&&&&+&+&?&*&?&?&?&?&?&?&+\\
\hline
34&-&&&&-&-&&-&&-&&?&?&-&-&+&?&?&?&+\\
\hline
36&&&&&&&&&&&&?&?&?&?&+&?&?&+&+\\
\hline
38&-&&&&-&-&&-&&-&&+&?&-&-&+&?&?&+&?\\
\hline
40&&&&&&&&&&&&&?&?&?&?&?&?&?&+\\
\hline
42&-&&&&-&-&&-&&-&&&+&-&-&+&?&*&?&+\\
\hline
44&&&&&&&&&&&&&&+&?&?&?&?&?&?\\
\hline
46&-&&&&-&-&&-&&-&&&&-&-&?&?&?&?&+\\
\hline
48&&&&&&&&&&&&&&&+&?&?&?&?&?\\
\hline
50&-&&&&-&-&&-&&-&&&&-&-&+&?&?&*&+\\
\hline
52&&&&&&&&&&&&&&&&&?&?&?&+\\
\hline
62&-&&&&-&-&&-&&-&&&&-&-&&&+&?&+\\
\hline
72&&&&&&&&&&&&&&&&&&?&?&+\\
\end{array}
$$
}
\label{table:2}
\end{table}

\section{Conclusion}\label{section:conclusion}
In this correspondence, we have constructed new families of optimal $q$-ary Euclidean self-dual codes from algebraic curves. With the same spirit, constructing more families of Euclidean self-dual codes from genus zero and genus one curves (over $\F_q$ with $q$ a prime) is worth considering. Characterization and constructions of Hermitian self-dual codes from algebraic geometry codes are also valuable.


%



\end{document}